\documentclass{article}

\usepackage{amsmath}
\usepackage{mathtools}
\usepackage{amsmath, amssymb, amsthm}
\usepackage{xspace}
\usepackage[dvipsnames]{xcolor}
\usepackage[nocompress,noadjust]{cite}
\usepackage{graphicx}
\usepackage{url}
\usepackage{tikz}
\usepackage{authblk}

\usetikzlibrary{decorations.pathmorphing, decorations.pathreplacing, decorations.shapes}

\makeatletter
\usetikzlibrary{decorations,backgrounds}
\def\pgfdecoratedcontourdistance{0pt}
\pgfset{
  decoration/contour distance/.code=%
    \pgfmathsetlengthmacro\pgfdecoratedcontourdistance{#1}}
\pgfdeclaredecoration{contour lineto closed}{start}{%
  \state{start}[
    next state=draw,
    width=0pt,
    persistent precomputation=\let\pgf@decorate@firstsegmentangle\pgfdecoratedangle]{%
    \pgfpathmoveto{\pgfpointlineattime{.5}
      {\pgfqpoint{0pt}{\pgfdecoratedcontourdistance}}
      {\pgfqpoint{\pgfdecoratedinputsegmentlength}{\pgfdecoratedcontourdistance}}}%
  }%
  \state{draw}[next state=draw, width=\pgfdecoratedinputsegmentlength]{%
    \ifpgf@decorate@is@closepath@%
      \pgfmathsetmacro\pgfdecoratedangletonextinputsegment{%
        -\pgfdecoratedangle+\pgf@decorate@firstsegmentangle}%
    \fi
    \pgfmathsetlengthmacro\pgf@decoration@contour@shorten{%
      -\pgfdecoratedcontourdistance*cot(-\pgfdecoratedangletonextinputsegment/2+90)}%
    \pgfpathlineto
      {\pgfpoint{\pgfdecoratedinputsegmentlength+\pgf@decoration@contour@shorten}
      {\pgfdecoratedcontourdistance}}%
    \ifpgf@decorate@is@closepath@%
      \pgfpathclose
    \fi
  }%
  \state{final}{}%
}
\makeatother
\tikzset{
  contour/.style={
    decoration={
      name=contour lineto closed,
      contour distance=#1
    },
    decorate}}

\newtheorem{theorem}{Theorem}
\newtheorem{claim}{Claim}
\newtheorem{corollary}{Corollary}
\newtheorem{lemma}{Lemma}

\newcommand{\Oh}{{\ensuremath{\mathcal{O}}}}

\newcommand{\prob}[1]{\textsc{#1}\xspace}

\theoremstyle{plain}

\newcommand{\ceil}[1]{\left\lceil #1 \right\rceil}

\newcommand{\etal}{\textit{et al.}\@}

\newcommand{\opt}{\textsc{opt}\xspace}
\newcommand{\algo}{\textsc{alg}\xspace}

\newcommand{\parent}{\textrm{parent}}

\title{On the Extended TSP Problem}

\author[1,2]{Juli\'{a}n Mestre}
\author[1]{Sergey Pupyrev}
\author[2]{Seeun William Umboh}

\affil[1]{Facebook Inc., USA}
\affil[2]{University of Sydney, Australia.}

\date{}

\begin{document}

\maketitle

\begin{abstract}
    We initiate the theoretical study of \prob{Ext-TSP}, a problem that originates in the area of profile-guided binary optimization. Given a graph $G=(V, E)$ with positive edge weights $w: E \rightarrow R^+$, and a non-increasing  discount function $f(\cdot)$ such that $f(1) = 1$ and $f(i) = 0$ for $i > k$, for some parameter $k$ that is part of the problem definition. The problem is to sequence the vertices $V$ so as to maximize $\sum_{(u, v) \in E} f(|d_u - d_v|)\cdot w(u,v)$, where $d_v \in \{1, \ldots, |V| \}$ is the position of vertex~$v$ in the sequence.

    We show that \prob{Ext-TSP} is APX-hard to approximate in general and we give a $(k+1)$-approximation algorithm for general graphs and a PTAS for some sparse graph classes such as planar or treewidth-bounded graphs.

    Interestingly, the problem remains challenging even on very simple graph classes; indeed, there is no exact $n^{o(k)}$ time algorithm for trees unless the ETH fails. We complement this negative result with an exact $n^{O(k)}$ time algorithm for trees.
\end{abstract}

\section{Introduction}

Profile-guided binary optimization (PGO) is an effective technique in modern compiles to improve performance by optimizing how binary code is laid out in memory. At a very high level, the idea is to collect information about typical executions of an application and then use this information to re-order how code blocks are laid out in the binary to minimize instruction cache misses, which in turn translates into running time performance gains. Newell and Pupyrev \cite{NewellP20} recently introduced an optimization problem, which they call the Extended TSP (\prob{Ext-TSP}) problem that aims at maximizing the number of block transitions that do not incur a cache miss.

The input to the \prob{Ext-TSP} problem is a weighted directed graph $G=(V, E)$, which in the context of PGO corresponds to the control flow representation of the code we are trying to optimize: Every node $u \in V$ corresponds to a basic block of code (for the purposes of this paper we can think of each of these blocks as a single instruction that takes a fixed amount of memory to encode); every edge $(u, v) \in E$ represents the possibility of an execution jumping from $u$ to $v$, and the weight $w(u, v)$ captures how many times the profiler recorded said jump during the data collection phase. Our ultimate goal is to find a linear ordering of the nodes, each of which represents a possible code layout of the binary; we let this linear ordering be encoded by a one-to-one function $d:V \rightarrow \{1, \ldots, |V|\}$. Finally, each edge $(u, v)$ contributes $f(|d_u - d_v|) \cdot w(u,v)$ to the objective, where $|d_u - d_v|$ is the distance between the edge endpoints in the linear ordering, and $f(\cdot)$ is a non-increasing discount function such that $f(1) = 1$ and $f(i) = 0$ for $i > k$, where $k = O(1)$ is part of the problem definition. 

Newell and Pupyrev~\cite{NewellP20} designed and evaluated heuristics for \prob{Ext-TSP} leading to significantly faster binaries. Their implementation is available in the open source project Binary Optimization and Layout Tool (BOLT) \cite{bolt,bolt-paper, NewellP20}. In their experiments, they found that setting $k$ to be a small constant\footnote{To be more specific, $k$ is the number of blocks that can fit into 1024 bytes of memory.} and $f(|d_u - d_v|) = \left(1 - \frac{|d_u - d_v|}{k}\right)$ for $1 < |d_u - d_v|< k$, yields the best results. The high level intuition is that the discount factor is a proxy for the probability that taking the jump causes a cache miss. Thus,  the \prob{Ext-TSP} objective aims at maximizing the number of jumps that do not cause a cache miss.

In this paper we initiate the theoretical study of \prob{Ext-TSP} by providing a variety of hardness and algorithmic results for solving the problem both in the approximate and the exact sense in both general and restricted graph classes.

\subsection{Our results}

We show that \prob{Ext-TSP} is APX-hard to approximate in general. We give a polynomial time $(k+1)$-approximation algorithm and a $n^{O(k/\epsilon)}$ time $(2+\epsilon)$-approximation for general graphs. We also give a $n^{O(k/\epsilon)}$ time $(1+\epsilon)$-approximation for some sparse graphs classes such as planar or treewidth-bounded graphs.

Interestingly, the problem remains challenging even on very simple graph classes; indeed, there is no exact $n^{o(k)}$ time algorithm for trees unless the ETH fail. Finally, we complement this negative result with an exact $n^{O(k)}$ time algorithm for trees.

\subsection{Related work}

PGO techniques have been studied extensively in the compiler's community. Code re-ordering is arguably the most impactful 
optimization among existing PGO techniques~\cite{bolt-paper}. The classical approach for code layout is initiated by 
Pettis and Hansen~\cite{PH90}, who formulated the problem of finding an ordering of basic blocks as a variant of the 
maximum directed \prob{TRAVELING SALESMAN PROBLEM} on a control flow graph. They describe two greedy heuristics 
for positioning of basic blocks. Later, one of the heuristics (seemingly producing better results) has been adopted by 
the community, and it is now utilized by many modern compilers and binary optimizers, 
including LLVM and GCC. Very recently, Newell and Pupyrev \cite{NewellP20} extended the classical model and suggested a new
optimization problem, called \prob{Extended-TSP}. With an extensive evaluation of real-world and synthetic applications,
they found the objective of \prob{Ext-TSP} is closely related to the performance of a binary; thus, an improved solution
of the problem yields faster binaries. 
We refer to \cite{NewellP20} for a complete background on this literature.

The problem of laying out data in memory to minimize the cache misses has been studied in the Algorithms community \cite{AchlioptasCN00, FiatKLMSY91, McGeochS91, SleatorT85}. In this setting a number of requests arrives online and our job is to design an eviction policy~\cite{Young2016}. Even though ultimately, we are also concerned with minimizing cache misses, there are two main differences: first, the profile data gives us information about future request that we can exploit to improve locality; second, this optimization is done at the compiler, which does not have control over the operating system's cache eviction policy. The benchmark used for online algorithms is the competitive ratio: the number of cache misses incurred by the online algorithm divided by the number of cache misses incurred by an optimal algorithm that knows the entire sequence of requests in advance. It is known that the best competitive ratio is $\Theta(k)$ for deterministic algorithms and is $\Theta(\log k)$ \cite{AchlioptasCN00}, where $k$ is the size of the cache. 

There are many classical optimization problems that seek for to sequence the vertex set of a graph to optimizing some objective function. The two most closely related to our problem are \prob{Max TSP} and \prob{Min Bandwidth}.

An instance of \prob{Max TSP} consists of a weighted undirected graph and our objective is sequence the vertex set to maximize the weight of adjacent nodes. The problem is known to be APX-hard \cite{PapadimitriouY93} and a number of approximation algorithms are known \cite{FisherNW79,KosarajuPS94,HassinR98,Serdyukov, HassinR00,ChenOW05,PaluchMM09,DudyczMPR17}, with the best being the $5/4$-approximation of Dudycz~\etal~\cite{DudyczMPR17} that runs in $O(n^3)$. 

An instance of \prob{Min Bandwidth} consists of an undirected graph and our objective is to sequence the vertex set to minimize the maximum distance between the endpoints of any edge in the graph. The problem admits an $n^{O(b)}$ time exact algorithm \cite{Sax80}, where $b$ is the bandwidth of the graph. On the negative side, there is no exact $g(b)n^{o(b)}$ time algorithm \cite{DL14} and unless the ETH fails, even in trees of pathwidth at most two. Several polylogarithmic approximation algorithms exist for different graphs classes \cite{Feige00,FeigeT09,Gupta01}; on the other hand, it is NP-hard to approximate the problem within any constant even for caterpillars \cite{DubeyFU11}.

A somewhat related problem is the \prob{Min Linear Arrangement} problem. An instance consists of an undirected graph and our objective is to sequence the vertex set to minimize the sum of the distances between the endpoints of each edge in the graph. Minimizing this objective function is equivalent to maximizing the \prob{Ext TSP} objective function with the discount function $f(i) = 1 - i/n$. \prob{Min Linear Arrangement} admits polynomial-time exact algorithms on trees \cite{Shiloach79}; however, we are not aware of any results for higher treewidth. There are several polylogarithmic approximation algorithms \cite{EvenNRS00, RaoR04, CharikarHKR10, FeigeL07} based on the spreading metrics technique of Even et al.~\cite{EvenNRS00}; however it is unclear how these techniques can be made to work for \prob{Ext TSP}. Moreover, for our applications, we are interested in the regime where $k \ll n$, so this connection does not yield a result of practical relevance.

\section{Problem definition and hardness}

An instance of \prob{Ext-TSP} problem consists of a directed graph $G=(V, E)$ with positive edge weights $w: E \rightarrow R^+$ and a non-increasing  discount function $f(\cdot)$ where $f(1) = 1$ and $f(i) = 0$ for $i > k$, where $k$, where $k$ is a parameter that is part of the problem definition. The problem is to sequence the vertices $V$ so that $d_v \in \{1, \ldots, |V| \}$ is position of vertex $v$ with the objective to maximize

\[ \sum_{(u, v) \in E} f(|d_u - d_v|) \cdot w(u, v) \]

The first thing to notice is that the fact that we could have defined the problem on an undirected graph since the contribution of an edge $(u,v)$ to the objective only depends on its weight and the distance between its two endpoints, and is independent of whether it is a forward or a backward jump. Indeed, we can reduce the undirected case to the directed case and vice versa: Given an undirected graph, we can orient the edges arbitrarily; while given a directed graph we can combine pairs of anti-parallel edges into a single edge by adding up their weight. 

In order to simplify our exposition, from now on we assume the input graph is undirected. Right away, this allows us to relate \prob{Ext-TSP} to \prob{Max TSP} and \prob{Min Bandwidth}, which in turn yields the following hardness results.

\begin{theorem}
\label{thm:hardness}
    The \prob{Ext-TSP} problem exhibits the following hardness:
    \begin{enumerate}
        \item it is APX-hard, even when $k=1$,
        \item does not admit an exact $n^{o(k)}$ time algorithm unless the ETH fails, even in trees.
    \end{enumerate}
\end{theorem}

\begin{proof}
    For the first part, we use the relation to \prob{Max TSP}. Recall that the objective of the latter problem is to maximize $\sum_{(u, v) \in E: |d_u - d_v| = 1} w_{(u, v)}$ given an undirected graph. We can reduce an instance of \prob{Max TSP} to an undirected instance of \prob{Ext-TSP} with $k= 1$ where $f(1) = 1$ and $f(2) = 0$. Therefore, \prob{Ext-TSP} is APX-hard even when $k=1$.

    For the second part, we use the relation to \prob{Min Bandwidth}. Recall that the objective of the latter problem is to minimize $\max_{(u, v) \in E} |d_u - d_v|$, the optimal value of this objective is called the \emph{bandwidth} of the graph. Given an instance $G$ with bandwidth $b$, consider  the \prob{Ext-TSP} instance where $f(i) = 1$ for $0 \leq i \leq k$ and $f(k+1) = 0$; if $k= b$ then the objective of this instance must be $w(E)$ as there exists a sequencing where the endpoints of every edge are within at most $k$ of one another. It follows that, if we could have an $n^{o(k)}$ time algorithm for \prob{Ext-TSP} that implies an $n^{o(b)}$ time algorithms, which does not exist even for very simple trees unless the ETH fails \cite{DL14}. 
\end{proof}

\section{Exact Algorithms}

In this section we complement the hardness from the previous section by developing an exact algorithm for trees whose running time is polynomial when $k = \Oh(1)$.

\begin{theorem}
    There is an $n^{O(k)}$ time algorithm for solving \prob{Ext-TSP} optimally on trees.
\end{theorem}

\begin{proof}
    Let $T$ be the input tree. Consider an optimal solution \opt, and let $O$ be the set of \emph{realized edges}, that is, the subset of edges whose endpoints are at distance at most $k$ in \opt. Without loss of generality we assume that each connected component of $O$ is laid out in a contiguous stretch in the optimal sequencing. Using this simple insight, we use dynamic programming (DP) to build a solution for the connected component $C$ that has the root of the tree, and solve separately the subtree rooted at nodes that are not in $C$ but that have a parent in $C$; we call such nodes \emph{dangling} nodes of $C$ (see Figure~\ref{fig:dangling}). Without loss of generality, we assume that $|C| \geq k$. If $C$ happens to be smaller, we can guess the optimal sequencing for $C$ (there are only $n^{k-1}$ choices), solve separately the subproblems rooted at dangling nodes of $C$, and keep the best solution.

    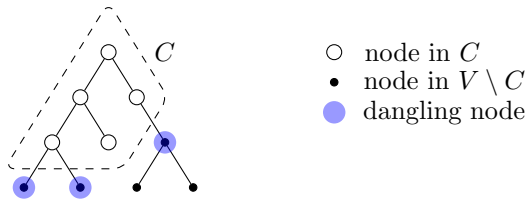
\begin{figure}
        \centering
        \begin{tikzpicture}[
            yscale=0.4,xscale=0.5,
            inC/.style={circle,draw=black,fill=white, inner sep=2pt},
            notC/.style={circle,draw=black,fill=black,inner sep=1pt},
            dangling/.style={circle,fill=blue,opacity=0.4,inner sep=3pt}            
            ]
            \node[inC] at (0,0) (c1) {} 
                child { node[inC] (c2) {}
                    child { node[inC] (c3) {} 
                        child {node[notC] (d1) {} edge from parent}
                        child {node[notC] (d2) {} edge from parent}
                    }
                    child { node[inC] (c4) {}
                    }
                }
                child { node[inC]  (c5) {}
                    child[missing] 
                    child { node[notC] (d3) {} edge from parent
                        child {node[notC]  {} edge from parent}
                        child {node[notC] {} edge from parent}
                    }
                }                
                ;
                
            \foreach \x in {1, 2, 3}
                \node[dangling] at (d\x) {};

            \node[inC,label={[label distance=5pt]right:{node in $C$}}] at (6,0) {};
            \node[notC,label={[label distance=6pt]right:{node in $V \setminus C$}}] at (6,-1) {};
            \node[dangling,label={[label distance=3pt]right:{dangling node}}] at (6,-2) {};

            \path[preaction={contour=-10pt,draw,dashed, rounded corners=4}] (c1.center) -- (c2.center) -- (c3.center) -- (c4.center) -- (c5.center) -- cycle;
            \node at (1.5, 0) {$C$};

        \end{tikzpicture}

        \caption{Dangling nodes of a root connected component $C$.}
        \label{fig:dangling}
    \end{figure}

    Our algorithm is based on a subtle DP formulation. Each DP state represents succinctly a partial solution for a subtree of $T$, and it is defined by a tuple $(z, \sigma, R)$, where
    \begin{itemize}
        \item $z \in V$ is the root of the subtree of $T$ we are trying to solve,
        \item $\sigma$ is a sequence of \emph{exactly} $k$ nodes in $T_z$, the subtree of $T$ rooted at $z$,
        \item $R$ is the set of edges incident on $\sigma$ that have already been realized.
    \end{itemize}

    It is worth noting that although the structure of the DP states builds on that used in the algorithm of Saxe~\cite{Sax80} for \prob{Min Bandwidth}, the fact that we do not necessarily realize all edges means we need new ideas and a more involved DP formulation to solve \prob{Ext-TSP}.

    Our high level goal is to build an edge weighted graph $H$ over these tuples plus two dummy source and sink nodes $s$ and $t$ such that every optimal solution to the \prob{Ext-TSP} problem on the subtree $T_z$ induces an $s$-$t$ path whose weight equals the value of this solution; and conversely, every $s$-$t$ path induces an \prob{Ext-TSP} solution of $T_z$ whose value equals the weight of the path. Thus, once the graph is defined and the equivalence established, solving \prob{Ext-TSP} amounts to a shortest path computation in $H$. 

    To provide some motivation and intuition on the definitions that will follow, consider an optimal solution of $T_z$ realizing a subset of edges $O$, where $C$ is the connected component of $(T_z, O)$ that contains the root $z$, and let $\tau$ be the optimal sequence for $C$. Note that $\tau$ realizes $O[C]$, and by our earlier assumption $|\tau| \geq k$. For each $j \in \{1, \ldots, |C| - k + 1\}$ we let $\sigma^j$ be the subsequence of $\tau$ from $j$ to $j + k -1$ and we let $R^j$ be the subset of edges realized by the first $j$ positions of $\tau$ that are incident on $\sigma^j$. Then the path induced by $\tau$ in $H$ will be 
    \[ s \rightarrow (z, \sigma^1, R^1) \rightarrow (z, \sigma^2, R^2) \rightarrow  \cdots \rightarrow (z, \sigma^{|C|-k + 1}, R^{|C|-k + 1}) \rightarrow t \]

    The weight of the first edge $s \rightarrow (z, \sigma^1, R^1)$ will be defined as the contribution of $\sigma^1$ to the objective, that is the total discounted (according to $\sigma^1$) weight of edges $R^1$. The weight of the last edge $(z, \sigma^{|C|-k + 1}, R^{|C|-k + 1}) \rightarrow t$ will be defined as the value of the subproblems defined by dangling nodes of $\sigma^{|C|-k + 1}$ not spanned by $R^{|C|-k + 1}$. Finally, the weight of an edge $(z, \sigma^j, R^j) \rightarrow (z, \sigma^{j+1}, R^{j+1})$ will be defined as the value of the subproblem defined by dangling nodes of $\sigma^j \setminus \sigma^{j+1}$ not spanned by $R^{j+1}$ plus the discounted weight of $R^{j+1} \setminus R^j$. Since we do not double count any contributions, the weight of the path adds up to the value of the optimal solution for $T_z$.

    Our goal is to impose some restrictions on the vertices and edges in $H$ so that every $s$-$t$ path induces a solution of equal value in $T_z$. To that end we will define the notion of valid tuples and valid edges, but before we do that, we must introduce a few more concepts.

    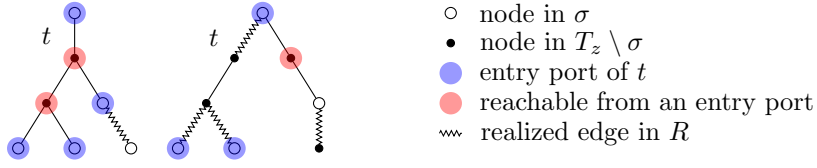
\begin{figure}
        \centering

        \begin{tikzpicture}[
            yscale=0.4,xscale=0.5,
            sigma/.style={circle,draw=black,fill=white, inner sep=1.5pt},
            T/.style={circle,draw=black,fill=black,inner sep=1pt},
            entry/.style={circle,fill=blue,opacity=0.4,inner sep=3pt},
            recheable/.style={circle,fill=red,opacity=0.4,inner sep=3pt}
            ]
            \node[sigma] at (0,0) (e1) {} 
                child { node[T,label={[label distance=4pt]160:$t$}] (r1) {}
                    child { node[T] (r2) {} 
                        child {node[sigma] (e2) {}}
                        child {node[sigma] (e3) {}}
                    }
                    child { node[sigma] (e4) {}
                        child[missing]
                        child {node[sigma] {} edge from parent[decorate,decoration={zigzag,segment length=2,amplitude=1}]}
                    }
                };
                
            \node[sigma] at (5,0) (e5) {} 
                child { node[T,label={[label distance=1pt]160:$t$}]  {}  edge from parent[decorate,decoration={zigzag,segment length=2,amplitude=1}]
                    child { node[T] {}  edge from parent[solid]
                        child {node[sigma] (e6) {} edge from parent[decorate,decoration={zigzag,segment length=2,amplitude=1}]}
                        child {node[sigma] (e7) {} edge from parent[decorate,decoration={zigzag,segment length=2,amplitude=1}]}
                    }
                    child[missing] 
                }
                child {node[T] (r3) {}
                    child[missing]
                    child {node[sigma] {}
                        child {node[T] {} edge from parent[decorate,decoration={zigzag,segment length=2,amplitude=1}]}
                        }
                    }
                ;

            \foreach \x in {1, 2, 3, 4, 5, 6, 7}
                \node[entry] at (e\x) {};

            \foreach \x in {1, 2, 3}
                \node[recheable] at (r\x) {};

            \node[sigma,label={[label distance=5pt]right:{node in $\sigma$}}] at (10,0) {};
            \node[T,label={[label distance=6pt]right:{node in $T_z \setminus \sigma$}}] at (10,-1) {};
            \node[entry,label={[label distance=3pt]right:{entry port of $t$}}] at (10,-2) {};
            \node[recheable,label={[label distance=3pt]right:{reachable from an entry port}}] at (10,-3) {};
            \draw [decorate,decoration={zigzag,segment length=2,amplitude=1}]
            (9.7, -4) -- +(0.6, 0) node [label={[label distance=0pt]right:{realized edge in $R$}}] {};

        \end{tikzpicture}

        \caption{
            \label{fig:entry-ports}
            Two example showing the entry ports of a node $t \in T_z \setminus \sigma$. On the left, all entry ports of $t$ are open, while on the right all entry ports of $t$ are closed.
        }
    \end{figure}

    Given a tuple $(z, \sigma, R)$ we say that a node $u \in \sigma$ is an \emph{entry port} for a node $t \in T_z \setminus \sigma$ if the unique path $P$ from $t$ to $u$ in $T$ does not go through any other vertex in $\sigma$; furthermore, we say that $u$ is a \emph{closed entry port} of $t$ if the edge in $P$ out of $u$ is in $R$, otherwise, we say $u$ is an \emph{open entry port} of $t$. Finally, we say that $t \in T_z \setminus \sigma$ is \emph{reachable} if all the entry ports of $t$ are open. See Figure~\ref{fig:entry-ports} for an example illustrating these definitions.

    A tuple $(z, \sigma, R)$ is \emph{valid} if for every $t \in T_z \setminus \sigma$ the entry ports $u \in \sigma$ for $t$ are either all closed or all open. Indeed if $(z, \sigma, R)$ was part of the path induced by some $\tau$ then either $t$ comes before $\sigma$ in $\tau$, in which case $t$ subtree spanned between the entry ports of $t$ must have been already realized; or $t$ comes after $\sigma$ in $\tau$, in which case said subtree will be realized later on. Thus, we can focus only on valid tuples. We define a graph $H$ over the valid tuples where we put a directed edge $(z, \sigma, R) \rightarrow (z, \sigma', R')$ if:
    \begin{itemize}
        \item $\sigma'$ is obtained from $\sigma$ by appending a reachable  node (reachable with respect to the first tuple) $v$ to $\sigma$ and removing the first node $u$ in $\sigma$,
        \item $R'$ equals $R$ minus edges in $R$ that are incident on $u$ but not on any other node in $\sigma$, plus edges from $v$ to $\sigma$,
        \item $(u, \parent(u)) \in R \cup R'$, 
        \item for each child $c$ of $u$ such that $(c, u) \notin R \cup R'$, $u$ is the unique (open) entry port of $c$ (defined with respect to the first tuple) and $v \notin T_c$; we call such $c$, a \emph{dangling child} of $u$.
    \end{itemize}

    Furthermore, we define the weight of such an edge to be the discounted weight of newly realized edges (namely, $R' \setminus R$) plus the total value of the optimal solutions for subtrees defined by dangling children of $u$. Note that the $R' \setminus R$ must connect $v$ to other nodes in $\sigma$, so we have all the information needed to discount their weight.

    Finally, we connect $s$ to each tuple $(z, \sigma, R)$ where $R$ is the set of edges with both endpoints in $\sigma$ and the weight of the edge is the discounted (w.r.t. $\sigma$) weight of $R$; and we connect each tuple $(z, \sigma, R)$ to $t$ if the only reachable nodes adjacent to $\sigma$ are dangling children and we set the weight of the edge to be the total value of the subproblems defined by those dangling children.
    
    Given a path $P$ in $H$ we define $\tau$ to be the induced solution by taking the $\sigma$ of the first tuple in the path, and then extending the ordering by appending the new node of the $\sigma$ in the next tuple and so on. Similarly, we can define the inverse operation: Given a sequencing $\tau$ realizing a connected component of nodes that have the root of the tree, then we can define a sequence of tuples such that the sequence of tuples induces $\tau$.

    \begin{claim} 
        \label{claim:realized}
        Let $P$ be a sequence a path out of $s$ in $H$ inducing some ordering $\tau$. Then $\tau$ realizes exactly the union of all the $R$-sets in $P$.
    \end{claim}

    The claim is easy to prove by induction on the length of the sequence. If the sequence has only one tuple $(z, \sigma, R)$, then $\tau = \sigma$ and  $R$ is the set of edges realized by $\sigma$, so the claim follows. Otherwise, if the last two tuples are $(z, \sigma, R)$ and $(z, \sigma', R')$ and $v$ is the last node in $\tau$ then $R' \setminus R$ is the set of edges realized  by $\tau$ incident on $v$ and we can use induction to account for the rest.

    In order to prove the correctness of our dynamic programming formulation, we need to argue that every solution $\tau$ to the original problem induces a path a equivalent cost, and vice-verse.
    
    \begin{claim}
        Let $\tau$ be the sequence of nodes in the connected component $C$ of edges realized by the optimal solution $\opt$ having $z$. The sequence of tuples induced by $\tau$ forms a valid $s$-$t$ path whose weight equals 
        \[ \sum_{(u, v) \in T[C]} f(|d_u - d_v|) w(u,v) + \sum_{\mathclap{\substack{u \notin C \\ \parent(u) \in C}}} \opt[T_u], \]
        where $d_u$ is the position of $u$ in $\tau$.
    \end{claim}

    If the sequence is a path, then by Claim~\ref{claim:realized}, $\tau$ realizes precisely the union of the $R$-sets in the sequence, and the weight of the path is precisely as stated in the claim. It only remains to show that the sequence is indeed a path. Consider two consecutive tuples $(z, \sigma, R)$ and $(z, \sigma', R')$ along the sequence. Our goal is to show that there is an edge connecting them. The first two conditions of a valid edge definition hold by definition of the induced sequence of tuples. For the third condition, note that $(u, \parent(u))$ must be realized by $\tau$ and so $\parent(u)$ must occur within $k$ positions of $u$ so the edge must appear in $R \cup R'$ and the condition holds. For the fourth condition, if we let $c$ be a child of $u$ such that $(c, u) \notin R \cup R'$, we note that $\tau$ cannot realize this edge after $\sigma'$, so it must be the case that $v \notin T[c]$ (otherwise $v$ would be disconnected from the root in $C$) and that $c$ is dangling child of $u$ (otherwise $c$ has a descendant in $\sigma$ that would be disconnected from the root in $C$).


    \begin{claim}
        For a given $s$-$t$ path in $H$, let $\tau$ be the ordering induced by the path. Then the set of edges realized by $\tau$ forms a connected component $C$ that contains the root and the weight of the path equals
            \[ \sum_{(u, v) \in T[C]} f(|d_u - d_v|) w(u,v) + \sum_{\mathclap{\substack{u \notin C \\ \parent(u) \in C}}} \opt[T_u], \]
        where $d_u$ is the position of $u$ in $\tau$.
    \end{claim}

    By Claim~\ref{claim:realized}, $\tau$ realizes precisely the union of the $R$-sets in the sequence. For every $v \in \tau$ other than $z$, we argue that $(v, \parent(v))$ is realized by $\tau$. Indeed, let $(z, \sigma, R)$ be the last tuple that such that $v \in \sigma$. If $(z, \sigma, R)$ is not the last tuple, by the third existence condition on the edge to the next tuple guarantees that $(u, \parent(u))$ is realized. If $(z, \sigma, R)$ is the last tuple, by the existence condition on the edge to $t$, all reachable nodes adjacent to $\sigma$ are dangling, in particular $\parent(u)$ is not reachable. Therefore, since $(v, \parent(v))$ is realized for all $v$, using induction we get that $v$ must be connected all the way to the root with realized edges. Therefore the vertices in $\tau$ form a connected subtree containing the root $z$, and the set of realized edges is precisely this subtree.
    
    All this effort would be for naught, unless we could represent $H$ succinctly. Recall that every node in $H$ is a tuple $(z, \sigma, R)$; clearly, there are only $n$ choices for $z$ and only $n^k$ choices for $\sigma$; furthermore, for an edge to be in $R$, since $\sigma$ is a contiguous chunk of size $k$, they can only realize edges with connection to the previous $k$ nodes, thus, we can represent $R$ succinctly by listing those additional $k$ nodes. Overall, there are $n^{2k+1}$ edges in $H$; we can list the outgoing neighboring tuples in $O(n)$ time per tuple\footnote{We do not attempt to optimize this running time.}. Therefore, we can run Dijkstra in $O(n^{2k + 2})$ time and identify the connected component of $z$. Since this has to be done for every node in $T$, we gain an extra factor of $n$ for a running time of $O(n^{2k + 3})$.
\end{proof}




\section{Approximation Algorithms for special graph classes}

In this section, we shows that we can get very good approximations for special graph classes that go beyond trees.

\begin{theorem}
    \label{thm:tree-width}
    There is an $n^{O(\frac{k t}{\epsilon})}$ time $(1 + \epsilon)$-approximation for \prob{Ext-TSP} in graphs with a tree decomposition of tree-width $t$.
\end{theorem}

\begin{proof}
    Let $T$ be the tree decomposition of our input graph $G$ and let $h = \ceil{1 / \epsilon}$. To simplify the presentation of our algorithm we define an auxiliary problem, where the goal is to partition the vertex set into clusters of size at most $hk$ and order each part separately, the \prob{Ext-TSP} objective is computed for each part and summed up. If we let \opt be the value of the optimal solution for the original problem, we claim that \opt', the value of the optimal solution for the auxiliary problem is not much lower; more precisely, 
    \[ \opt' \geq \frac{h-1}{h} \opt. \]
    To see this, suppose that \opt lists the vertices in the order $v_1, v_2, \ldots, v_n$. We pick a random threshold $\alpha$ u.a.r. from $\{0, 1, \ldots, k-1\}$, and cluster vertices together so that for each $j$ we have a cluster $\{v_{hk j + 1 + \alpha}, \ldots, v_{hk (j+1) + \alpha} \}$, yielding a solution to the auxiliary problem. Note that the probability of an edge that is realized by $\opt$ must have endpoints that are at most $k$ apart in the ordering, so there is only a $1/h$ chance of that edge not being present in $\opt'$. Although this is a randomized construction, and it just shows that $E[\opt'] \geq \frac{h-1}{h} \opt$, it is easy to see that there must exist a value of $\alpha$ that yields the desired bound\footnote{Note that the argument is non-constructive in the sense that given $G$ it is not clear how to partition $G$ into clusters of size $hk$ so that $\opt' \geq \frac{h-1}{h} \opt$. The argument only guaranteed the existence of such a clustering.}.

    Given a tree decomposition for $G$ with treewidth $t$, and a bag $B$ in the decomposition we denote with $T[B]$ the subset of vertices in the original graph spanned by the sub-decomposition rooted at $B$. 
    For each $u \in B$ we define a collection orderings of subsets $\mathcal{S}_u$, such that for an ordering $\sigma$ of a subset $S \subseteq V$ of vertices to be in $\mathcal{S}_u$ we require that:
    \begin{itemize}
        \item $|S| \leq hk$,
        \item $u \in S$, and
        \item the subgraph $\big(S, \{ (a, b) \in E[S] : |\sigma(a) - \sigma(b)| \leq k \}\big)$ is connected.
    \end{itemize}
         
    We define a dynamic programming formulation for our auxiliary problem as follows. For each bag $B$ in the decomposition and each $|S|$-tuple $(\sigma_u : u \in B)$ where $\sigma_u \in \mathcal{S}_u$, we create a dynamic programming state $A[B, (\sigma_u : u \in B)]$ that corresponds to the cost of the best solution for $T[B]$ where each $\sigma_u$ is the ordering of one of the clusters in the solution of the auxiliary problem. To keep the requirements feasible we ask that for any $u, v \in B$ if $\sigma_u$ and $\sigma_v$ have one or more vertice in common then $\sigma_v = \sigma_u$.

    We work with a nice tree decomposition with join, forget, and introduce nodes. To define the recurrence for $A$ we consider each case. 
    \begin{itemize}
        \item Join node: Here we have children with the same bag as the node. We simply pass the tuple constraining the solution space to each child. To compute its value we add the value of the two children and subtract the contribution of edges inside of $B$ to avoid double counting. Notice that the distance between the endpoints of $E[B]$ is specified by $(\sigma_u : u \in B)$ so we can compute the appropriate discount of these edges.
        \item Introduce node: Here we have a single child with a bag having one fewer element; call it $u$. We remove $\sigma_u$ from the tuple and $u$ from $B$. To compute its value we add the contribution of edges between $u$ and other nodes in $\sigma_u$ to the value of the child. Again,  we can use $\sigma_u$ to discount the weight of these edges accordingly. 
        \item Forget node: Here we have a single child with a bag with one additional element, call it $u$. To compute its value we need to guess the $\sigma_u$ in the optimal solution. If $u$ happens to already be in the part of some other $\sigma_v$ of $v \in B$ then $\sigma_u = \sigma_v$. Otherwise, we must guess $\sigma_v$ by picking $hk$ vertices from $T[B] \setminus \cup_{v \in B} \sigma_v$ and checking that $\sigma_v \in \mathcal{S}_u$. Taking the best value state over all possibilities yields the value of the parent state.
    \end{itemize}
    
    For the correctness, notice that there is no loss of information in the case of a introduce node. Let $u$ be the node begin introduced. Either $u$ is the only vertex in common between $B$ and $\sigma_u$, in which case $u$ is the only vertex in $T[B]$ by virtue of $\sigma_u$ being connected in $G$, and so it is safe to forget $\sigma_u$ together with $u$ in the child node. Or, there exists another $v \in B - u$ such that $v \in \sigma_u$, which case $\sigma_v = \sigma_u$ and so the information about the constraints we imposed in $u$'s part are preserved further down the decomposition.

    For the correctness of the forget node case, note that the component that $u$ belong to in the optimal solution is connected and that $B$ acts like a separator from $T[B]$ to the rest of the graph, so if $u$ is not in the same component as any node in $B$, then it must be in a component with only nodes in $T[B]  \setminus \cup_{v \in B} \sigma_v$.
      
    There are $n^{hkt + 1}$ states in the decomposition and each one is considered once by a state associated with the parent bag in the decomposition, so the overall work is linear on the number of the states. We can enumerate the states on the fly by paying another $O(n)$ term per state so the total running time is $n^{hkt + 2}$.
    
    Now, setting $h = 1 + 1/\epsilon$, the optimal solution found by DP is bound to be a $1 + \epsilon$ approximation for the original problem in $n^{O(kt/\epsilon)}$ as promised in the Theorem statement.
\end{proof}

We can use this result to obtain a $(1+\epsilon)$-approximation for planar graphs.

\begin{corollary}
    There is an $n^{O(\frac{k}{\epsilon^2})}$ time $(1 + \epsilon)$-approximation for \prob{Ext-TSP} in planar graphs.
\end{corollary}

\begin{proof}
    Using Baker's technique \cite{Baker} we can find an $\ell$-outerplanar subgraph $G'$ of the input graph $G$ such that value of the optimal solution to the \prob{Ext-TSP} in $G'$ is at least $1- 2/\ell$ the value of the optimal solution in $G$. Since the treewidth of $G'$ is no more than $3\ell$, we can use the algorithm from Theorem~\ref{thm:tree-width} get a $1+\epsilon'$ approximation in $G'$ in $n^{O(k\ell/\epsilon')}$ time. Setting $\epsilon' = \epsilon/3$ and $\ell = 6/\epsilon$, we get the desired result for any $\epsilon \leq 1$.
\end{proof}

\section{Approximation Algorithms for general graph}

\subsection{Greedy}

Consider the following greedy algorithm: Start with an arbitrary vertex, and on each step append a vertex with the heaviest edge to the last-added vertex; i.e. if $u$ is the last-added vertex, then we append the vertex $u$ maximizing $w(u,v)$.

\begin{lemma}
    Greedy is a $2k$-approximation and this is tight. It can be implemented to run in $O(m \log n)$ time.
\end{lemma}

\begin{proof}
    
    Let $O$ be the edges realized by the optimal solution and let $u_1, u_2, \ldots, u_n$ be the order computed by the greedy algorithm. Let $d^*_u$ be the position of $u$ in the optimal solution. Observe that the value of the greedy solution is at least $\sum_{i=1}^{n-1} f(1)w(u_i,u_{i+1}) = \sum_{i=1}^{n-1} w(u_i, u_{i+1})$ as $f(1) = 1$. We partition $O$ as follows, for each $u_i$ we have a part $O_i = \{ (u_i, u_j) \in O : j > i \}$. Using the fact that $f$ is non-increasing and the definition of the greedy algorithm, $f(|d^*_u - d^*_v|)w(u,v) \leq w(u,v) \leq w(u_i, u_{i+1})$ for all $(u,v) \in O_i$, and $|O_i| \leq 2k$. Thus, the value of the optimal solution is
    \[ \sum_{i=1}^{n-1} \sum_{(u,v) \in O_i} f(|d_u - d_v|) w(u,v) \leq 2k \sum_{i=1}^{n-1} w(u_i, u_{i+1}). \]
    Thus, greedy is a $2k$-approximation. 

    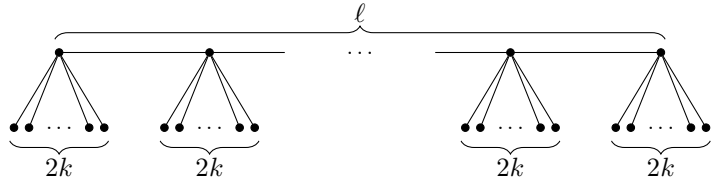
\begin{figure}
        \centering
        \begin{tikzpicture}[vertex/.style={circle,draw=black,fill=black,inner sep=1pt},xscale=2]
            \foreach \x in {0, 1, 3, 4} {
                \draw (\x, 0) node[vertex] (r\x) {};
                \draw (r\x) -- ++(-0.3, -1) node[vertex] (left) {};
                \draw (r\x) -- ++(-0.2, -1) node[vertex] {};
                \draw (r\x) -- ++(0.2, -1) node[vertex] {};
                \draw (r\x) -- ++(0.3, -1) node[vertex] (right) {};
                \draw (r\x) + (0,-1) node {\footnotesize $\ldots$};

                \draw [
                    decorate,
                    decoration={mirror,brace,amplitude=5pt,raise=5pt}
                ] (left.west) -- (right.east)
                  node [black,midway,yshift=-15pt] {$2k$};
            }
            \draw(0, 0) -- +(1, 0);
            \draw(1, 0) -- +(0.5, 0) node {};
            \draw(3, 0) -- +(1, 0);
            \draw(3, 0) -- +(-0.5, 0) node {};
            \draw(2, 0) node {\footnotesize $\ldots$};

            \draw [
                decorate,
                decoration={brace,amplitude=5pt,raise=5pt}
            ] (r0.west) -- (r4.east)
                node [black,midway,yshift=15pt] { $\ell$};

        \end{tikzpicture}

        \caption{Tight instance for greedy. Optimal solution can realize $2k\ell$ edges while Greedy may end up realizing only $\ell -1 + k$ edges.}
    \end{figure}

    To show that the analysis is tight, consider the following instance with $n= (2k+1) \ell$ consisting on $\ell$ $2k$-stars with the centers of the stars connected with a path of length of length $\ell-1$. All edges have weight 1. The discount function $f$ is such that $f(i) = 1$ when $i \leq k$ and $f(i) =0$ when $i > k$. The optimal solution sequences one star after the other and achieves a total cost of $2k \ell$. While the greedy solution may start at the center of the "left most star" and traverse the centers of all star and then add $k$ pendant nodes, achieving a total cost fo $\ell - 1 + k$. By making $\ell$ large we get an approximation ratio that tends to $2k$.


  For the implementation, we need to maintain a maximum priority queue with the nodes that are yet to be added to the greedy solution. The value associated with node $u$ is the weight of the edge connecting $u$ to the last node in the current partial greedy solution. When a new node is added to the greedy solution, this causes the priority of certain vertices to be updated (up for those incident on $u$ or down for those incident on the second last-node of the partial solution, or either direction if incident on both nodes). The key observation is that each edge can cause the priority of a node to be changed twice (once when the first endpoint is added to the solution and again when that endpoint stops being the last node of the greedy solution). Therefore, the total number of priority updates is $O(m)$, which using a simple binary heap yields the desired time.
\end{proof}

\subsection{Cycle cover based algorithm}

We can do slightly better if we use a maximum weight cycle cover as the basis for our solution. A similar approach has been used to design approximation algorithms for max-TSP~\cite{FisherNW79}.

\begin{theorem}
\label{thm:cycle-cover}
    There is a polynomial time $\left(1 + \frac{1}{k+1} \right)k$-approximation for \prob{Ext-TSP} in general graphs.
\end{theorem}

\begin{proof}
    Let $A$ be a maximum weight set of edges such that the degree of every node is at most 2. This problem is also known as maximum weight simple $2$-matching and can be reduced to regular maximum weight matching~\cite[Ch.~30]{book/combopt/Schrijver}. Note that $A$ is a collection of paths and cycles in $G$. If there exists a cycle $C$ in $A$, we break $C$ by removing the lightest edge. This gives us a collection of paths $A'$. Sequencing each path gives a solution to the \prob{Ext-TSP} problem with value at least $w(A')$.

    Now, given a solution to the \prob{Ext-TSP} problem with value $\opt$, we claim that we can construct a solution to the degree bounded problem that has value at least $\opt / k$. To see this, note that the weight of the edges whose endpoints are at distance \emph{exactly} $i$ for $i=1, \ldots, k$ is a candidate solution for $A$. It follows then that $w(A) \geq \opt /k$.
    
    This is because the edges that are counted towards the objective in \prob{Ext-TSP} have maximum degree $2k$ and that solution can be scaled down by a factor of $k$ to get a fractional solution to an exact LP formulation of the degree bounded problem. Thus, we get that $w(A) \geq \opt/k$.

    Let $\algo$ be the value of the solution found by our algorithm. Consider a cycle $C$ in $A$ with length $\ell = |C|$. Let $e$ be the edge in $C$ with minimum weight. Therefore, $C$ contributes at least $w(C) - w(e) + f(\ell - 1) w(e)$ to $\algo$. Since $w(e) \leq w(C) / \ell$, we can further simplify the previous expression to
    \[ w(C) \left(1 - \frac{1 - f(\ell -1)}{\ell} \right) \]
    
    Now if each cycle $C$ in $A$ had length at least $\ell > k + 1$, the weight of $C \cap A'$ would be at least $w(C) \left( 1 - \frac{1}{k +2} \right)$. Let $\algo$ be the cost of the solution found by our algorithm. Then 
    \[ \algo \geq w(A') \geq \frac{k+1}{k + 2} w(A) \geq \frac{k+1 }{(k+2) k} \opt, \]
    which matches the approximation factor of $k+1$ promised in the theorem statement. Unfortunately, cycles can be as small as $\ell = 3$, which depending on $f$ could yield a worse approximation factor, so we need a different approach to our analysis.

    Let $\ell^*$ be the number in $[3, 4, \ldots, k+2]$ maximizing $\frac{1 - f(\ell^* -1)}{\ell}$. Using the same reasoning as above, we see that
    \[ \algo \geq w(A) \left(1 - \frac{1 - f(\ell^* -1)}{\ell^*} \right). \]
    The first thing to note is that if $\ell^* = k+2$ then the above analysis yield the desired approximation, so from now one assume $\ell^* < k + 2$ and $\frac{1 - f(\ell^* -1)}{\ell^*} > \frac{1}{k+2}$, or equivalently, that \[1 - \frac{\ell^*}{k+2} > f(\ell^*-1).\]

    Consider the edges realized by the optimal solution and split them into $X$ and $Y$. The first set, $X$, are the edges whose endpoints are at distance at most $\ell^*-2$ from each other; the second set, $Y$, are the edges whose endpoints are at distance between $\ell^*-1$ and $k$. Notice that
    \[ \opt \leq w(X) + f(\ell^* - 1) w(Y), \]
    since all edges in $Y$ are discounted at least $f(\ell^*-1)$, and that
    \[ w(A) \geq \max \left\{ \frac{w(X)}{\ell^*-2}, \frac{w(Y)}{k - \ell^* -2} \right\},\]
    since we can use the same scaling argument on $X$ or $Y$ but using a smaller scaling factor since the vertices in those edges sets have smaller degrees; namely, $2\ell^*-2$ and $2(k - \ell^* - 1)$ respectively. Putting the above two inequalities together we get
    \begin{align*}
        \opt 
        & \leq (\ell^* -2) w(A) + f(\ell^* -1) (k - \ell^* -2) w(A) \\
        & \leq \frac{(\ell^* -2) + f(\ell^* -1)  (k - \ell^* -2) } {\left(1 - \frac{1 - f(\ell^* -1)}{\ell^*} \right)} \algo. 
    \end{align*}

    Think of the above upper bound on the approximation ratio $\opt/ \algo$ as a function of $f(\ell^*-1)$. We want to find the value $ 0 \leq f(\ell^*-1) \leq 1 - \ell^*/(k+2)$ that yields the worst bound on the approximation ratio. The upper bound is the ratio of two linear functions of $f(\ell^*-1)$ and is thus maximized when either $f(\ell^*-1) = 0$ or $f(\ell^*-1) = 1 - \ell^*/(k+2)$. If $f(\ell^*-1) = 0$, the ratio simplifies to $\frac{\ell^* - 2}{1 - 1/\ell^*}$, which in turn is maximized at $\ell^* = k + 2$ and yields a ratio of $k + \frac{k}{k+1}$, as desired. Finally, if $f(\ell^*-1) = 1 - \ell^*/(k+2)$, we again get the same approximation ratio.
\end{proof}

\subsection{Local search algorithm}

So far all the algorithms we have presented in this section have polynomial running times that are independent of $k$. If we are willing to have algorithms that run in $n^{O(k)}$ we can get arbitrarily good approximations.

Our local search algorithm is parameterized by an integer value $\ell \geq k$. The algorithm maintains a solution $\tau$ and performs local search moves where some subset of $\ell$ nodes are taken out of $\tau$ and sequenced optimally and attached to the end of the solution. At each step we perform the best such move and we stop once there is no move that improves the solution.

\begin{lemma}
\label{lem:local}
    A local optimal solution is a $2 + \frac{2}{\ell/k - 1}$ approximation for the \prob{Ext-TSP} in general graphs.
\end{lemma}

\begin{proof}
    We will use the following notation throughout this proof: For a given solution $\tau$ and a permutation $\sigma$ of $\ell$ elements, let $\tau - \sigma$ the permutation of $n-k$ elements that we get by removing the nodes in $\sigma$ from $\tau$. Also, let $\tau | \sigma$ be the permutation obtained by concatenating $\sigma$ to $\tau - \sigma$. Finally, let $w_\sigma (\tau)$ be the discounted weight of edges realized by $\tau$ that are incident on vertices in $\sigma$, and $w(\tau)$ be the discounted weight of all edges realized by $\tau$, i.e. the value of $\tau$.

    Assume that $\tau$ is locally optimal; namely, that no local move can improve its value:
    \[ w(\tau) \geq w(\tau | \sigma) \quad \forall \sigma : |\sigma| = \ell. \]

    Notice that $w(\tau) \leq w(\tau - \sigma) + w_\sigma(\tau)$ and that $w(\tau | \sigma) \geq w(\tau- \sigma) + w(\sigma)$. Therefore, a weaker necessary condition for being locally optimal is that
    \[ w_\sigma(\tau) \geq w(\sigma) \quad \forall \sigma : |\sigma| = \ell. \]

    Let us build a a collection for $n + \ell$ sub-sequences of the optimal solution by sliding a window of size $\ell$ over $\opt$. Call the resulting collection $S$. Adding up the above inequality for all $\sigma \in S$ we get
        \[ \sum_{\sigma \in S} w_\sigma(\tau) \geq \sum_{\sigma \in S} w(\sigma)  \]

    Notice that every edge realized by $\tau$ can appear in at most $2\ell$ terms in the left-hand side of the above inequality (this is because every endpoint appears in at most $\ell$ permutations), while every edge realized by $\opt$ must appear in at least $\ell-k$ terms in the right-hand side of the above inequality. These observation imply the following relation between $\tau$ and $\opt$
    \[ 2\ell w(\tau) \geq (\ell-k) w(\opt), \]
    which in turn finish off the proof of the lemma.
\end{proof}

Of course, the issue with the above algorithm is that it is not clear how to compute a locally optimal solution. However, we can use the usual trick of only making a move if it improves the value of the objective by at least $\delta/n w(\opt)$. This guarantees that we do not perform more than $n/\delta$ and degrades the approximation ratio by no more than $2 \delta$. This yields an algorithm that runs in $O(n^{\ell + 1} / \delta)$ time.

\section{Conclusions and open problems}

Some generalizations are easy to handle: when the discount function $f$ is non-symmetric, when the block sizes are non-uniform. There are a few interesting questions that remain unanswered:
\begin{enumerate}
    \item Is there an $O(1)$-approximation in polynomial time, independent of $k$?
    \item Is there an exact $\Oh(f(k, t) n^{O(k)})$ time algorithm where $t$ is the treewidth of the instance?
    \item Is there an $O(c^n)$ time algorithm where $c > 1$ is some constant?
\end{enumerate}




Note that we cannot expect $(1+\epsilon)$-approximation even in $n^{O(k)}$ time since that would contradict APX-hardness of \prob{Max TSP}, and we cannot expect to get exact algorithms for bounded treewidth instances in $n^{o(k)}$ time either due to \prob{Min Bandwidth} hardness.

\section*{Acknowledgments} We would like to thank Vahid Liaghat for fruitful discussions on the \prob{Ext-TSP} problem.

\bibliographystyle{plainurl}
\bibliography{refs}

\begin{thebibliography}{10}

\bibitem{bolt}
Binary optimization and layout tool, 2020.
\newblock URL: \url{https://github.com/facebookincubator/BOLT}.

\bibitem{AchlioptasCN00}
Dimitris Achlioptas, Marek Chrobak, and John Noga.
\newblock Competitive analysis of randomized paging algorithms.
\newblock {\em Theor. Comput. Sci.}, 234(1-2):203--218, 2000.

\bibitem{Baker}
Brenda~S. Baker.
\newblock Approximation algorithms for np-complete problems on planar graphs.
\newblock {\em Journal of the ACM}, 41(1):153–180, 1994.

\bibitem{CharikarHKR10}
Moses Charikar, Mohammad~Taghi Hajiaghayi, Howard~J. Karloff, and Satish Rao.
\newblock \emph{\emph{l}}\({}_{\mbox{2}}\)\({}^{\mbox{2}}\) spreading metrics
  for vertex ordering problems.
\newblock {\em Algorithmica}, 56(4):577--604, 2010.

\bibitem{ChenOW05}
Zhi{-}Zhong Chen, Yuusuke Okamoto, and Lusheng Wang.
\newblock Improved deterministic approximation algorithms for max {TSP}.
\newblock {\em Inf. Process. Lett.}, 95(2):333--342, 2005.

\bibitem{DL14}
Markus~Sortland Dregi and Daniel Lokshtanov.
\newblock Parameterized complexity of bandwidth on trees.
\newblock In {\em International Colloquium on Automata, Languages, and
  Programming}, pages 405--416. Springer, 2014.

\bibitem{DubeyFU11}
Chandan~K. Dubey, Uriel Feige, and Walter Unger.
\newblock Hardness results for approximating the bandwidth.
\newblock {\em J. Comput. Syst. Sci.}, 77(1):62--90, 2011.

\bibitem{DudyczMPR17}
Szymon Dudycz, Jan Marcinkowski, Katarzyna~E. Paluch, and Bartosz Rybicki.
\newblock A 4/5 - approximation algorithm for the maximum traveling salesman
  problem.
\newblock In {\em Proc of 19th International Conference on Integer Programming
  and Combinatorial Optimization}, volume 10328, pages 173--185, 2017.

\bibitem{EvenNRS00}
Guy Even, Joseph Naor, Satish Rao, and Baruch Schieber.
\newblock Divide-and-conquer approximation algorithms via spreading metrics.
\newblock {\em J. {ACM}}, 47(4):585--616, 2000.

\bibitem{Feige00}
Uriel Feige.
\newblock Approximating the bandwidth via volume respecting embeddings.
\newblock {\em J. Comput. Syst. Sci.}, 60(3):510--539, 2000.

\bibitem{FeigeL07}
Uriel Feige and James~R. Lee.
\newblock An improved approximation ratio for the minimum linear arrangement
  problem.
\newblock {\em Inf. Process. Lett.}, 101(1):26--29, 2007.

\bibitem{FeigeT09}
Uriel Feige and Kunal Talwar.
\newblock Approximating the bandwidth of caterpillars.
\newblock {\em Algorithmica}, 55(1):190--204, 2009.

\bibitem{FiatKLMSY91}
Amos Fiat, Richard~M. Karp, Michael Luby, Lyle~A. McGeoch, Daniel~Dominic
  Sleator, and Neal~E. Young.
\newblock Competitive paging algorithms.
\newblock {\em J. Algorithms}, 12(4):685--699, 1991.

\bibitem{FisherNW79}
Marshall~L. Fisher, George~L. Nemhauser, and Laurence~A. Wolsey.
\newblock An analysis of approximations for finding a maximum weight
  hamiltonian circuit.
\newblock {\em Oper. Res.}, 27(4):799--809, 1979.

\bibitem{Gupta01}
Anupam Gupta.
\newblock Improved bandwidth approximation for trees and chordal graphs.
\newblock {\em J. Algorithms}, 40(1):24--36, 2001.

\bibitem{HassinR98}
Refael Hassin and Shlomi Rubinstein.
\newblock An approximation algorithm for the maximum traveling salesman
  problem.
\newblock {\em Inf. Process. Lett.}, 67(3):125--130, 1998.

\bibitem{HassinR00}
Refael Hassin and Shlomi Rubinstein.
\newblock Better approximations for max {TSP}.
\newblock {\em Inf. Process. Lett.}, 75(4):181--186, 2000.

\bibitem{KosarajuPS94}
S.~Rao Kosaraju, James~K. Park, and Clifford Stein.
\newblock Long tours and short superstrings (preliminary version).
\newblock In {\em 35th Annual Symposium on Foundations of Computer Science,
  Santa Fe, New Mexico, USA, 20-22 November 1994}, pages 166--177. {IEEE}
  Computer Society, 1994.

\bibitem{McGeochS91}
Lyle~A. McGeoch and Daniel~Dominic Sleator.
\newblock A strongly competitive randomized paging algorithm.
\newblock {\em Algorithmica}, 6(6):816--825, 1991.

\bibitem{NewellP20}
Andy Newell and Sergey Pupyrev.
\newblock Improved basic block reordering.
\newblock {\em {IEEE} Transactions in Computers}, 69(12):1784--1794, 2020.

\bibitem{PaluchMM09}
Katarzyna~E. Paluch, Marcin Mucha, and Aleksander Madry.
\newblock A 7/9 - approximation algorithm for the maximum traveling salesman
  problem.
\newblock In {\em Proc of 12th International Workshop on Approximation,
  Randomization, and Combinatorial Optimization}, volume 5687, pages 298--311,
  2009.

\bibitem{bolt-paper}
Maksim Panchenko, Rafael Auler, Bill Nell, and Guilherme Ottoni.
\newblock Bolt: A practical binary optimizer for data centers and beyond.
\newblock In {\em Proceedings of the 2019 IEEE/ACM International Symposium on
  Code Generation and Optimization}, page 2–14, 2019.

\bibitem{PapadimitriouY93}
Christos~H. Papadimitriou and Mihalis Yannakakis.
\newblock The traveling salesman problem with distances one and two.
\newblock {\em Math. Oper. Res.}, 18(1):1--11, 1993.

\bibitem{PH90}
Karl Pettis and Robert~C Hansen.
\newblock Profile guided code positioning.
\newblock In {\em SIGPLAN Notices}, volume~25, pages 16--27. ACM, 1990.

\bibitem{RaoR04}
Satish Rao and Andr{\'{e}}a~W. Richa.
\newblock New approximation techniques for some linear ordering problems.
\newblock {\em {SIAM} J. Comput.}, 34(2):388--404, 2004.

\bibitem{Sax80}
James~B Saxe.
\newblock Dynamic-programming algorithms for recognizing small-bandwidth graphs
  in polynomial time.
\newblock {\em SIAM Journal on Algebraic Discrete Methods}, 1(4):363--369,
  1980.

\bibitem{book/combopt/Schrijver}
Alexander Schrijver.
\newblock {\em Combinatorial Optimization}.
\newblock Springer-Verlag, 2003.

\bibitem{Serdyukov}
A.~I. Serdyukov.
\newblock An algorithm with an estimate for the traveling salesman problem of
  maximum (in russian).
\newblock {\em Upravlyaemye Sistemy}, 25:80--86, 1984.

\bibitem{Shiloach79}
Yossi Shiloach.
\newblock A minimum linear arrangement algorithm for undirected trees.
\newblock {\em {SIAM} J. Comput.}, 8(1):15--32, 1979.

\bibitem{SleatorT85}
Daniel~Dominic Sleator and Robert~Endre Tarjan.
\newblock Amortized efficiency of list update and paging rules.
\newblock {\em Commun. {ACM}}, 28(2):202--208, 1985.

\bibitem{Young2016}
Neal~E. Young.
\newblock {\em Online Paging and Caching}, pages 1457--1461.
\newblock Springer New York, 2016.

\end{thebibliography}

\end{document}